\documentclass[5p]{elsarticle}
\usepackage{amsmath,amssymb}
\usepackage{soul,hyperref}
\usepackage{qtree}
\usepackage{mathtools}
\newtheorem{theorem}{Theorem}[section]
\newtheorem{definition}[theorem]{Definition}
\newtheorem{example}[theorem]{Example}

\newtheorem{corollary}[theorem]{Corollary}
\newtheorem{lemma}[theorem]{Lemma}
\newtheorem{proposition}[theorem]{Proposition}
\newtheorem{remark}[theorem]{Remark}
\hypersetup{
    colorlinks=true,
    linkcolor=blue,
    filecolor=blue,      
    urlcolor=blue,
    pdfpagemode=FullScreen,
    }
\newcommand{\error}{\operatorname{error}}
\DeclareMathOperator{\adom}{adom}
\newenvironment{proof}{\par\noindent\textbf{Proof.}}{\hfill$\Box$\par\bigskip\par}

\newcommand{\revnote}[1]{#1}
\journal{Information Processing Letters}
\begin{document}

\begin{frontmatter}

\title{On the non-efficient PAC learnability of conjunctive queries}

\author[1]{Balder ten Cate\texorpdfstring{\fnref{fn1}}{}}

\fntext[fn1]{Research supported by the European Union’s Horizon 2020 research and innovation programme (MSCA-101031081). We thank Victor Dalmau for helpful feedback on an early version of this note. }

\affiliation[1]{organization={ILLC, University of Amsterdam},
            addressline={Postbus 94242}, 
            city={Amsterdam},
            postcode={1090 GE}, 
            country={The Netherlands}}

\author[2]{Maurice Funk}

\affiliation[2]{organization={Leipzig University},
            addressline={Augustusplatz 10}, 
            city={Leipzig},
            postcode={04109}, 
            country={Germany}}

\author[3]{Jean Christoph Jung}

\affiliation[3]{organization={TU Dortmund University},
            addressline={August-Schmidt-Stra\ss{}e 1}, 
            city={Dortmund},
            postcode={44227}, 
            country={Germany}}

\author[2]{Carsten Lutz}

\begin{abstract}
This note serves three purposes: (i) we provide a  self-contained exposition of the fact that conjunctive queries are not efficiently learnable in 
the \emph{Probably-Approximately-Correct (PAC)} model,
paying clear attention to the  
complicating fact that 
this concept class lacks the \emph{polynomial-size fitting property}, a property that is tacitly assumed in much of the computational learning theory literature;
(ii) we establish a strong negative PAC learnability result that applies to many restricted classes of conjunctive queries (CQs), including acyclic CQs for a wide range of notions of acyclicity;
(iii) we show that CQs (and UCQs) are 
efficiently PAC learnable with membership queries.
\end{abstract}


\end{frontmatter}


\section{Introduction}\label{sec:intro}

Conjunctive queries (CQs) are an extensively studied database query language
that plays a prominent role in database theory. CQs correspond precisely to 
Datalog programs with a single non-recursive rule
and to the positive-existential-conjunctive fragment of first-order logic. 
Since the evaluation problem for conjunctive queries is NP-complete, various tractable subclasses have been introduced and studied. These include
different variants of acyclicity, such as $\alpha$-acyclicity, $\beta$-acyclicity, $\gamma$-acyclicity, and Berge-acyclicity, which
form a
strict hierarchy with Berge-acyclicity being most restrictive  \cite{Fagin83:acyclicity}. A landmark result by Grohe states
that a class of CQs is tractable if and only if the treewidth
of all CQs in it is bounded by a constant (under certain assumptions)~\cite{DBLP:journals/jacm/Grohe07,DBLP:journals/jacm/Marx13}.
%

In this note, we consider the learnability
of  CQs from labeled examples, in Valiant's well-known \emph{Probably Approximately Correct} (PAC) learning model~\cite{Valiant84:pac}.
We give a self-contained proof that the class of all CQs as well as all classes of acyclic CQs mentioned above are \emph{not} efficiently PAC learnable. While the
general idea of our proof is due to~\cite{DBLP:conf/ecml/Kietz93,Haussler89},
we strengthen the result in several respects and present it in a form that is easily
accessible to modern-day database theorists.

\begin{figure*}
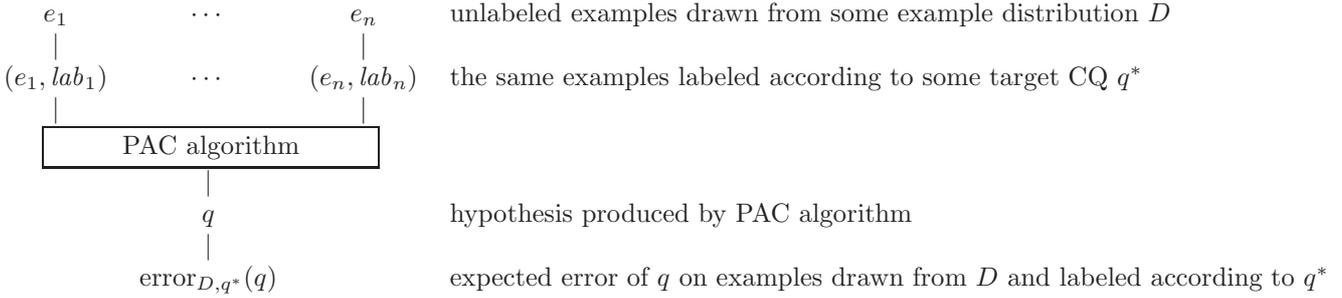

{
\centering
\begin{tabular}{cccl}
$e_1$ & $\cdots$ & $e_n$ & unlabeled examples drawn from some example distribution $D$ \\
$|$ && $|$ & \\
$(e_1,\mathit{lab}_1)$ & $\cdots$ & $(e_n,\mathit{lab}_n)$ & the same examples labeled according to some target CQ $q^*$ \\
$|$ && $|$ & \\
\multicolumn{3}{c}{\fbox{~~~~~~~ PAC algorithm ~~~~~~~}} \\
& $|$ & \\
& $q$ && hypothesis produced by PAC algorithm \\
& $|$ & \\
& $\error_{D,q^*}(q)$ && expected error of $q$ on examples drawn from $D$ and labeled 
according to $q^*$
\end{tabular}
}


\caption{Graphical depiction of a PAC algorithm}
\label{fig:pac}
\end{figure*}

The result $q(I)$
of evaluating a $k$-ary CQ $q$ on a database instance $I$ is a 
set of $k$-tuples of values from the active domain of~$I$. 
An \emph{example}, then, is most naturally taken to be a
pair $(I,\textbf{a})$ where $I$ is a database instance and 
$\textbf{a}$ is a $k$-tuple of values from the active domain of $I$. The example is \emph{positive}  
if $\textbf{a}\in q(I)$ and  \emph{negative} otherwise.

An \emph{efficient PAC algorithm} is a (possibly randomized) polynomial-time algorithm that takes as input a 
set of examples 
drawn from an unknown probability
distribution~$D$ and labeled as positive/negative
according to an unknown target CQ $q^*$ to be learned,
and that outputs a CQ $q$,
such that,
if the input sample is sufficiently large,
then with probability at least $1-\delta$, $q$ has expected
error at most $\epsilon$, meaning that if we draw an example $e$ from~$D$, then with probability $1-\epsilon$, $q$ and $q^*$ assign
the same label to $e$
(cf.~Figure~\ref{fig:pac}).
The required number of examples must furthermore be bounded by a
 function polynomial in $|q^*|$, $1/\delta$, $1/\epsilon$, and the example size.
We give a precise definition in Section~\ref{sec:preliminaries}.
Note that since a PAC algorithm does not know the example distribution $D$, it must perform well
for \emph{all} distributions $D$.
 In this sense, the PAC  model captures a strong
form of distribution-independent learning.

Our main result is the following, stated, for simplicity, for
unary CQs:
\begin{theorem} \label{thm:main} (assuming $\text{RP} \neq \text{NP}$)
Let $C$ be any class of unary CQs over a fixed
schema $\mathbf{S}$ that contains at least one binary relation symbol and one unary relation symbol. 
If $C$ includes
all path-CQs, then $C$
    is not efficiently PAC learnable, even w.r.t.~single-instance example distributions.
\end{theorem}

Here,  RP denotes the class of problems 
solvable by a randomized algorithm with one-sided error that runs in polynomial time, and
by a \emph{path-CQ} we mean a unary CQ 
of the form
\[ \begin{array}{@{}l@{}l}
  q(x_1) \coloneq \exists x_2 \ldots x_n (&R(x_1,x_2)\land \cdots \land R(x_{n-1},x_n) \\[1mm]
  & \land P(x_{j_1})\land\cdots\land P(x_{j_m}))
  \end{array}
\]
where $R$ is a binary relation symbol and $P$ is a unary relation symbol. That is, a path-CQ is a very simple type of CQ that describes
an outgoing directed path decorated with a single unary relation symbol.

With a \emph{single-instance example distribution}, we mean
an example distribution $D$ such that for some database instance $I$, $D$ assigns non-zero probability
mass only to examples of the form $(I,\textbf{a})$. 
This captures 
the natural scenario of learning CQs from positive and 
negative examples that all pertain to a single given database instance.
Clearly, efficient PAC learnability w.r.t.~all example distributions implies efficient PAC learnability w.r.t.~single-instance distributions.

%
%
%
Note that efficient PAC learnability is 
not an anti-monotone property of query classes, and  Theorem~\ref{thm:main} says more than just that path-CQs are not efficiently PAC learnable.
In particular, Theorem~\ref{thm:main} implies that \emph{the class of all CQs} is not efficiently PAC learnable, and the same is true for all classes of acyclic CQs mentioned above since path-CQs belong to all of these classes. Theorem~\ref{thm:main} also implies  
non-efficient PAC learnability of concept expressions
in the description logics $\mathcal{EL}$ and $\mathcal{ELI}$ (even in the absence of a TBox), see e.g.~\cite{Funk2019:when} and references therein.

It is worth comparing the notion of a \emph{PAC learning algorithm} to that of a \emph{fitting algorithm}. 
Both types of algorithms take as input a set of labeled examples. A fitting 
algorithm decides the existence of a CQ that agrees with the labels of the input examples.
The fitting problem is coNExpTime-complete
for CQs~\cite{Willard10,CateD15} and, in fact, is known to be hard already for some more restricted classes of acyclic CQs~\cite{CateD15,funk2019:msc, Funk2019:when}.
A PAC algorithm, on the other hand, produces a CQ that, with high probability,
has a low expected error, but is not required to fit the input examples.
Despite these differences, it is well-known that for concept classes that are both polynomial-time evaluable and have the polynomial-size fitting property (defined in Section~2), NP-hardness of the fitting problem implies the non-existence of an efficient PAC learning algorithm~\cite{PittValiant88}, see Proposition~\ref{prop:pac-vs-fitting} below.
Unfortunately, the concept class of CQs has neither of these properties. A main 
difficulty of our proof of Theorem~\ref{thm:main} (which is nevertheless
based on a reduction from an NP-hard fitting problem) 
is to find a way around this.

We also prove that PAC learnability of CQs can be recovered by extending 
the PAC model with membership queries, known from Angluin's~\cite{Angluin88} 
model of exact learning. In a membership query, the learner chooses an example $(I,\mathbf{a})$ and asks an oracle to provide, in unit time,
the positive or negative labeling of $(I,\mathbf{a})$ according to the target query.
In Angluin's model of exact learning, CQs are known to not be efficiently learnable 
with membership queries alone, but they are efficiently learnable when also
equivalence queries are admitted (the learner may give a hypothesis query
to the oracle and ask whether it is equivalent to the target query, 
requesting a counterexample if this is not the case). The latter is implicit in~\cite{CateDK13:learning}, an explicit proof can be found in~\cite{tCD2022:conjunctive}, cf.~also~\cite{Arias2002:learning}.


As pointed out in~\cite{CateDK13:learning}, the fact that
CQs are efficiently exactly learnable with membership and equivalence
queries implies PAC learnability with membership queries and
an NP-oracle (cf.~\cite{Angluin88}), where the NP-oracle is used for evaluating hypotheses on examples.
It was left open whether CQs are efficiently PAC learnable with membership queries \emph{without} an NP-oracle. We give an affirmative answer to this question
and show that it also extends to UCQs, that is, to disjunctions of conjunctive
queries.
\begin{theorem}\label{thm:main2}
Fix any schema $\mathbf{S}$ and $k\geq 0$. The class of all $k$-ary CQs over $\mathbf{S}$ is efficiently PAC learnable with membership queries. The same is true for the class of all $k$-ary UCQs over $\mathbf{S}$.
\end{theorem}

\subsection{Related work}
\label{sec:relwork}

Haussler~\cite{Haussler89} shows that the class of Boolean CQs over a schema that contains
an unbounded number of unary relation symbols is not efficiently PAC-learnable
(unless $\text{RP}=\text{NP}$). The essential part of the proof is to show that  
the fitting problem for the same concept class is NP-complete. Over a schema that consists
of unary relation symbols only, every CQ is trivially Berge-acyclic. Therefore, this 
  implies that efficient PAC learnability fails for acyclic Boolean CQs, for
any of the aforementioned notions of acyclicity. 
The fact that Haussler's result is stated for Boolean CQs and Theorem~\ref{thm:main} is stated for unary CQs is an inessential difference (cf.~\cite{DBLP:conf/ecml/Kietz93}).
The fact that the proof in~\cite{Haussler89} uses an unbounded number of unary relation symbols, however, is an important difference. Indeed, if one was to 
consider Boolean queries over a fixed finite schema that consists of unary relation symbols only, then the resulting concept class would be finite and trivially PAC learnable. 

Kietz~\cite{DBLP:conf/ecml/Kietz93} proves that the class of unary CQs over a schema that contains a single binary relation symbol and an 
unbounded number of unary relation symbols is not PAC-learnable (unless $\text{RP}=\text{NP}$). Again the essential part of the proof is to show that  
the fitting problem is NP-complete. Kietz's result already applies to path-CQs of length~1 with multiple unary relation symbols.
This is only possible because of the infinite schema, as, otherwise,
the concept class is again finite and trivially PAC learnable.

\revnote{
Cohen \cite{DBLP:conf/colt/Cohen96} proves that the class of
unary CQs over a schema that contains two binary relation 
symbols is not PAC-predictable unless certain assumptions from
the field of cryptography fail.
In \emph{PAC prediction}, 
the output of the algorithm is not required to be a concept from the concept class, but instead must be any polynomial-time evaluable concept such
as a polynomial-time algorithm. 
PAC learnability implies PAC predictability 
for concept classes that are polynomial-time evaluable
(cf.~Remark~\ref{rem:prediction-vs-learning}).
%
Cohen's result already applies to
path-CQs (defined slightly differently than above, using two binary relation symbols and no
unary relation symbol -- this difference is inessential).
As a consequence, Cohen's result yields the restriction of Theorem~\ref{thm:main} to polynomial-time evaluable classes $C$
(such as the class of all acyclic CQs, under any of the mentioned notions of
acyclicity), under cryptographic assumptions.
Moreover, in contrast to PAC learnability, PAC predictability is an anti-monotone property of concept classes. Thus, Cohen's result also yields
Theorem~\ref{thm:main} for efficient PAC predictability in place of     
efficient PAC learnability, again under cryptographic assumptions.

In an earlier paper \cite{Cohen93}, Cohen had proved a related but weaker result that 
requires relation symbols of arity three. The work of Hirata~\cite{Hirata2005:prediction}, in a similar vain, 
shows that there is even a fixed database on which efficient
PAC prediction (and thus also learning) of acyclic CQs is impossible -- a stronger condition than single-instance example distributions. The result, however, requires ternary relation symbols and CQs of unbounded arity.
We also remark that it follows from general results of Schapire, see Section~6.3 of~\cite{DBLP:journals/ml/Schapire90}, that any class of CQs that is NP-hard to evaluate is not efficiently PAC-predictable unless $\text{NP} \subseteq \text{P/poly}.$
}

We consider, in this note,  classes of CQs defined through
acyclicity conditions. In the literature on inductive logic programming (ILP) various positive and negative
PAC learnability results have been obtained for classes of CQs
defined by different means (e.g., limitations on the use of existential variables, determinacy conditions pertaining to functional relations, and restricted variable depth). These are orthogonal to acyclicity. An overview can be found in \cite[Chapter 18]{ChengWolf:1997}.

In~\cite{CateDK13:learning}, the authors study learnability of 
\emph{GAV schema mappings}, which are closely related to
\emph{Unions of Conjunctive Queries (UCQs)}.
Specifically, it was proved in \cite{CateDK13:learning} that GAV schema
mappings are not efficiently PAC learnable, assuming  $\text{RP} \neq \text{NP}$, on
source schemas that contain at least one relation symbol of arity at
least two, using a reduction of the non-PAC-learnability of
propositional formulas in positive DNF. This result immediately
implies that, for any schema $\mathbf{S}$ containing a relation
symbol of arity at least two, and for each $k\geq 0$,
the class of $k$-ary UCQs over $\mathbf{S}$ is not efficiently PAC learnable, assuming  $\text{RP} \neq \text{NP}$. Additionally, in~\cite{CateDK13:learning}, the authors completely map out the (non-)learnability of restricted classes of UCQs definable by conditions on their Gaifman graph.

There is also another line of work on PAC learnability of conjunctive queries \cite{Dalmau1999,DalmauJeavons2003,ChenValeriote2019} that is somewhat different in nature: one fixes a schema $\mathbf{S}$ and an $\mathbf{S}$-instance $I$ and defines a concept class where the concepts are now all relations over the active domain of $I$ definable by a $k$-ary CQ (as evaluated in $I$).
PAC learning for various classes of Boolean formulas, such as
3-CNF, can be seen as a special case of this framework, for 
a specific choice of schema $\mathbf{S}$ and (two-element) instance $I$, where $k$ then
corresponds to the number of Boolean variables.
Since, for a fixed choice of $k$,
this yields a finite concept class,
in this setting, one is interested
in the complexity of PAC learning as a function of $k$.
The mentioned papers  establish effective dichotomies, showing that, depending on the 
choice of $\mathbf{S}$ and $I$,  
this concept class is either efficiently PAC learnable in $k$ or is not even efficiently PAC predictable with membership queries in $k$ (under suitable cryptographic assumptions).  See also Remark~\ref{rem:prediction-vs-learning} below.

\section{Preliminaries}
\label{sec:preliminaries}

\subsection{Conjunctive Queries}

A \emph{schema} $\mathbf{S}$ is a finite set of relation symbols with associated arity. 
An \emph{instance} $I$ over schema $\mathbf{S}$ is 
a finite set of facts over~$\mathbf{S}$, where a \emph{fact}
is an expression of the form $R(a_1, \ldots, a_n)$
where $R\in\mathbf{S}$ is an $n$-ary relation symbol and $a_1,\dots,a_n$ are \emph{values}.
The \emph{active domain} of an instance $I$,
denoted by $\adom(I)$ is the 
(finite) set of values that occur in the facts of $I$.

A $k$-ary \emph{conjunctive query (CQ)} over a schema $\mathbf{S}$, for $k\geq 0$, is an expression of 
the form $$q(\textbf{x}) \coloneq \exists\textbf{y}(\alpha_1\land\cdots\land \alpha_n)$$
where $\textbf{x},\textbf{y}$ are tuples of variables, $\textbf{x}$ has length $k$, and each conjunct $\alpha_i$ is an atomic
formula that uses a relation symbol from $\mathbf{S}$ and only variables from $\textbf{x}$ and $\textbf{y}$ , such that
each variable from $\textbf{x}$ occurs in some conjunct.
We denote by $q(I)$ the set of all $k$-tuples
$\textbf{a}$ such that 
$I\models q(\textbf{a})$. 


We will not define in depth the various notions of 
acyclicity that have been mentioned in the introduction, but
we reiterate here that they form a hierarchy with Berge-acyclicity
being most restrictive, and that all mentioned classes of
acyclic queries are polynomial-time evaluable, meaning that 
given a CQ $q(\textbf{x})$ from the class, an instance $I$ and a tuple $\textbf{a}$ of elements of the active domain of~$I$, we can 
decide in polynomial time whether $\textbf{a}\in q(I)$.

The definition of \emph{path-CQs} was  given in Section~\ref{sec:intro}.

\begin{example}
An example of a path-CQ is the query 
\[ q(x) \coloneq \exists yzu (R(x,y)\land R(y,z)\land R(z,u)\land P(y)\land P(u)).\]
\end{example}

Every path-CQ is Berge-acyclic and hence polynomial-time evaluable, see
e.g.\ the classic paper where this is proved for 
$\alpha$-acyclic queries~\cite{Fagin83:acyclicity}.



\subsection{Computational Learning Theory}

A \emph{concept class} is a triple
$C=(\Phi, \mathit{Ex}, \models)$, where $\Phi$ is a 
set of concepts, $\mathit{Ex}$ is a set of examples,
and ${\models}\subseteq \mathit{Ex}\times \Phi$ represents whether
an example is a positive or a negative example for a given
concept. We also denote by $\mathit{lab}_\phi(e)$ the \emph{label}
of $e$ according to $\phi$, that is, $\mathit{lab}_\phi(e)=+$ if
$e\models\phi$ and $\mathit{lab}_\phi(e)=-$ otherwise.
Two concepts $\phi,\phi'\in \Phi$
are said to be \emph{equivalent} if 
$\mathit{lab}_\phi(e)=\mathit{lab}_{\phi'}(e)$ for all $e\in \mathit{Ex}$.%
\footnote{
This deviates slightly from the standard convention, which defines a concept class to be a 
pair $(\mathit{Ex}, C)$ where $C\subseteq \wp(\mathit{Ex})$
(and, for $c\in C$, $|c|$ to be the
size of the smallest representation of $c$). The difference is non-essential.
We prefer this presentation
as it makes it easier to spell out unambiguously the 
algorithmic problems that we consider (e.g., Definition~\ref{def:poly-eval})
}

A \emph{labeled example} is a pair $(e,s)$ with $e\in \mathit{Ex}$ and
$s\in \{+,-\}$. A concept $\phi\in\Phi$ \emph{fits} a set of labeled examples $E$ if $\mathit{lab}_\phi(e)=s$ for all $(e,s)\in E$.

We  only consider countable concept classes.
Concepts and examples are assumed to have an effective
representation and a corresponding notion of size,
which is denoted by $|\phi|$ and $|e|$, respectively.
We  also denote the set of all concepts (examples) of size at
most $n$ by $\Phi_{(n)}$ (respectively, $Ex_{(n)}$).
For a finite set of (possibly labeled) examples $E$, 
$||E|| = \sum_{e\in E}|e|$.

The following two properties of concept classes will be important
for us later on:

\begin{definition}[Polynomial-time evaluability]
\label{def:poly-eval}
A concept class is \emph{polynomial-time evaluable}
if there exists a polynomial-time algorithm that, given 
$\phi\in \Phi$ and $e\in \mathit{Ex}$, outputs a Boolean indicating whether
$e\models\phi$.
\end{definition}

\begin{definition}[Polynomial-size fitting property]
A concept class has the \emph{polynomial-size fitting property}
if for every finite set of labeled examples $E$, the
existence of a concept that fits $E$ implies that there exists
a fitting concept whose size is bounded by a polynomial in 
$||E||$.
\end{definition}

We now define the two algorithmic problems mentioned in the introduction, namely fitting and PAC learning.

\begin{definition}[Fitting problem]
The \emph{fitting problem} (also known as \emph{consistency problem} or
\emph{separability problem}) for a concept class $C$ is the problem to
decide, given a finite set of labeled examples $E$, whether there exists a
concept in~$C$ that fits~$E$. 
\end{definition}

In order to define PAC algorithms, we first need to introduce some terminology and
notation. An \emph{example distribution} for a concept class $C=(\Phi,\mathit{Ex},\models)$
is a probability distribution $D$ over $\mathit{Ex}$.
Given concepts $\phi, \phi^*\in \Phi$ and an example distribution $D$, 
\[\error_{D,\phi^*}(\phi)= \mathop{\textrm{Pr}}_{e \in D}(\mathit{lab}_\phi(e) \neq \mathit{lab}_{\phi^*}(e))\]
is the expected
error of $\phi$ relative to $\phi^*$ and $D$.

\begin{definition}[Efficient PAC learnability]
\label{def:pac}
An \emph{efficient PAC algorithm} for a concept class $C$ is a pair $(A,f)$ where
\begin{itemize}
    \item $A$ is a
randomized polynomial-time algorithm
that takes as input a set of labeled examples and outputs a concept from $C$, and
\item $f(\cdot,\cdot,\cdot,\cdot)$ is a polynomial function, such that, 
    for all $\delta,\epsilon\in (0,1)$,  all $n,m\in\mathbb{N}$, all
    example distributions $D$ over $\mathit{Ex}_{(m)}$, and  all $\phi^*\in \Phi_{(n)}$, if the input consists of at least $f(1/\delta,1/\epsilon,n,m)$
    examples drawn from $D$ and labeled according to $\phi^*$, then with probability at least $1-\delta$, 
    $A$ outputs a concept $\phi$ with $\error_{\phi^*,D}(\phi) \leq \epsilon$. 
\end{itemize}
  If such an algorithm exists, we say that $C$ is 
  \emph{efficiently PAC learnable}.
  If the function $f$ depends only on $\delta$ and
  $\epsilon$ and not on $n, m$, then we say that $(A,f)$ is
  a \emph{strongly efficient} PAC algorithm, and that
  the concept class $C$ is 
  \emph{strongly efficiently PAC learnable}.
\end{definition}

This definition of efficient PAC algorithms is modeled after the one in the textbook~\cite{AnthonyBiggs92}, in line with the literature on inductive logic programming (cf., e.g.,~\cite{ChengWolf:1997}). 
Our results also apply
to the alternative oracle-based definition.\footnote{\label{footnote:oracle}
Following the oracle-based
presentation in, e.g.,~\cite{KearnsVazirani}, one can define an 
efficient PAC learning algorithm for a concept
class $C$ to be a randomized polynomial-time algorithm that takes as input $\delta,\epsilon \in (0,1)$ and a bound $n \in \mathbb{N}$ on the size of the target concept~$\phi^*$, and that has access to an oracle $\mathit{EX}_{\phi^*,D}$ which, when called, returns
(in unit time) a random example drawn from~$D$
and labeled according to $\phi^*$. For every
choice of  $\delta, \epsilon$, $\phi^*\in \Phi$, $n\geq |\phi^*|$, 
and for every example distribution~$D$, the algorithm must terminate in time polynomial in 
$1/\delta, 1/\varepsilon$, $n$, and 
the size of the largest example returned by the 
oracle.
Furthermore, it must 
return a concept that with probability $1-\delta$ satisfies $\error_{\phi^*,D}(\phi)<\epsilon$.

Note that, under this 
definition, 
not only the running time of the algorithm but also the number of 
examples drawn from the distribution may depend 
on the size of examples:
if the learning algorithm
encounters a large example $e$, it may follow up by requesting a number of additional
examples that is polynomial in the size of $e$.

Efficient PAC learnability in the above sense implies efficient PAC learnability in the sense of
Definition~\ref{def:pac}: \revnote{one can 
turn an oracle-based learning algorithm into a learning
algorithm according to Definition~\ref{def:pac} 
by drawing examples uniformly at random from the input batch to answer EX oracle calls. (To guarantee polynomial-time termination, even on inputs where a 
fitting concept does not exist, we can maintain a counter and terminating after $p(n)$ steps, where
$p$ is the polynomial that bounds the running time 
of the oracle-based learner on consistent inputs).}
Our negative learnability results thus apply also to the
oracle-based definition. A classic paper that shows equivalence of different
PAC learning models is \cite{DBLP:journals/iandc/HausslerKLW91}.}
%
%
We prefer the  above definition
as it exhibits more clearly the relationship to fitting algorithms.

The following proposition relates the two algorithmic problems (fitting and PAC learning) to 
each other. 

\begin{proposition}[Pitt and Valiant~\cite{PittValiant88}]
\label{prop:pac-vs-fitting}
Let $C$ be a polynomial-time evaluable concept class with the poly\-nomial-size fitting property.
If $C$ is efficiently PAC learnable, then the fitting problem 
for $C$ is in RP.
\end{proposition}

This is a well-known fact (cf.~also~\cite[Thm~6.2.1]{AnthonyBiggs92}), although not in this precise formulation, as, usually,
polynomial evaluability and the polynomial-size fitting property are tacitly 
assumed (\revnote{which has sometimes led to mistakes,
e.g., in the derivation of Corollary~15 in \cite{DBLP:conf/ecml/Kietz93}}).
To be self-contained,
we outline the proof of Proposition~\ref{prop:pac-vs-fitting} here.

\smallskip

\begin{proof} (of Proposition~\ref{prop:pac-vs-fitting})
Assume that there is an efficient PAC algorithm $(A,f)$ for $C$.
We use it to solve the fitting problem for $C$ in randomized polynomial time.
Assume that a set $E$ of $k$ labeled examples is given as the
input. 
Let $n=p(||E||)$, where $p$ is the polynomial
witnessing the fact that $C$ has the polynomial-size fitting property.
Let $D$ be the uniform distribution on $E$ (where each example in $E$ gets probability mass $1/k$), and let $m$ be the maximum size of an 
example in $E$.
Pick $\delta < .5$ and $\epsilon < 1/k$.
We generate a new (polynomial-sized) collection of labeled examples $E'$ by
drawing $f(1/\delta,1/\epsilon,n,m)$ samples from distribution $D$, and run algorithm $A$ on it.
Finally, we check that the output of $A$ is a fitting concept for $E$. If so, we answer Yes. Otherwise,
we answer No.

Clearly, if there is no fitting concept, the output will be No. If, on the other hand, there is a fitting concept, then there is one of size at most $n$, and
   hence, with probability $1-\delta$, the algorithm will 
output a concept
   with error less than $\epsilon$. This in fact implies that the error is $0$
   (because if the query misclassifies an example to which $D$ assigns non-zero mass, then it will have error at least $1/k$).
   Hence, with probability $1-\delta > 0.5$ the algorithm outputs Yes.
\end{proof}

A variation on the same argument shows:

\begin{proposition}\label{prop:strong-learnability}
If a concept class is strongly efficiently PAC learnable,
then it has the polynomial-size fitting property.
\end{proposition}

\begin{proof} The proof uses the same construction as before, 
except that the sample size now does not depend on $n$.
Furthermore, we omit the verification step where we 
confirm that the produced concept fits the input examples. Instead, we just output the result of the learning algorithm. 
In this way, we obtain a randomized polynomial-time algorithm that has a non-zero probability of outputting a fitting concept for given input labeled examples, 
whenever a fitting concept exists.
The polynomial-size fitting property immediately 
follows from this (the run that outputs a fitting concept does
so in polynomial time).
\end{proof}

We also make use of the following trivial fact:

\begin{proposition}\label{prop:subspace}
    If a concept class $(\Phi,\mathit{Ex},\models)$ is efficiently PAC learnable, then, for every $\mathit{Ex}'\subseteq \mathit{Ex}$, the 
    concept class $(\Phi,\mathit{Ex}',\models)$ is also efficiently PAC learnable.
\end{proposition}

Indeed, this follows from the fact that every example distribution over $\mathit{Ex}'_{(n)}$ 
is in particular also an example distribution over $\mathit{Ex}_{(n)}$
(that assigns no probability mass to any example in 
$\mathit{Ex}\setminus \mathit{Ex}'$).

Finally, we   use  a well known connection between  PAC algorithms and \emph{Occam algorithms}.

\begin{definition}[Occam algorithm]\label{def:occam}
    An \emph{Occam algorithm} for a concept class $C=(\Phi,\mathit{Ex},\models)$, 
with parameters $\alpha<1$ and $k\geq 1$, is an algorithm that takes as input a set of labeled examples $E$  and outputs a concept $\phi\in \Phi$ with $|\phi|\leq {|E|}^{\alpha} |\phi^*|^k$ that fits $E$ provided that any 
concept from $\Phi$ does. Furthermore, the running time is required to be bounded by a polynomial in $|\phi^*|$ and $||E||$. 
\end{definition}
 
Blumer \textit{et al.}~\cite{Blumer89:learnability} proved that every Occam algorithm $A$ yields an efficient PAC algorithm, namely $A'=(A,f)$, where 
the sample-size polynomial $f$ is chosen such that
 $$f(1/\delta,1/\epsilon,n,m)=\left(\frac{n^k\ln 2+\ln(2/\delta)}{\epsilon} \right)^{1/(1-\alpha)}.$$
Note that  $f$ does not depend 
on its fourth component $m$ (i.e., the example size bound).
Moreover, every Occam algorithm gives rise to an efficient PAC algorithm, not only in the sense of Definition~\ref{def:pac} as explained above, but, by the same arguments, also when considering the oracle-based presentation of PAC algorithms (cf.~Footnote~\ref{footnote:oracle}).

\begin{theorem}[\cite{Blumer89:learnability}] \label{thm:occam}
Every concept class for which there is an Occam algorithm
is efficiently PAC learnable.
\end{theorem}

\section{Classes of CQs as Concept Classes}

Each class of CQs can be naturally viewed as a
concept class. Fix a schema $\mathbf{S}$, an arity $k\geq 0$,
and a class $C$ of $k$-ary CQs over $\mathbf{S}$.
In the associated concept class 
$(C, \mathit{Ex}, \models)$,
$\mathit{Ex}$ is the class of all 
pairs $(I,\textbf{a})$ with $I$ an $\mathbf{S}$-instance
and $\textbf{a}$ a $k$-tuple of elements of the active domain of $I$, and $\models$ describes query answers, that is, $(I,\textbf{a})\models q(\textbf{x})$ iff $\textbf{a}\in q(I)$, for all $q(\textbf{x})\in C$ and $(I,\textbf{a})\in \mathit{Ex}$.
We may abuse notation and refer to this concept
class $(C, \mathit{Ex}, \models)$ simply as $C$ when no 
ambiguity arises. The following theorem summarizes some
basic properties.

\begin{theorem}[\cite{CateD15,Cate2022:extremal}] \label{thm:cq-negative}
Fix any schema $\mathbf{S}$ 
that contains at least one binary relation symbol, and some $k\geq 0$.
\begin{enumerate}
    \item The concept class of $k$-ary CQs over $\mathbf{S}$ is not polynomial-time evaluable (unless $\text{P}=\text{NP}$).
Indeed, its evaluation problem is NP-complete.%
\footnote{\revnote{The evaluation problem takes as input $\phi$ and $e$ and asks if $e\models\phi$.}}
\item The concept class of $k$-ary CQs over $\mathbf{S}$ lacks the polynomial-size fitting property. Indeed, the smallest fitting CQ for a given set of labeled examples is in general exponentially large.
    \item The fitting problem for $k$-ary CQs over $\mathbf{S}$ is coNExpTime-complete.
\end{enumerate}
\end{theorem}

%
Let us now consider restricted classes of (unary) CQs that still include path-CQs.
We will see in the next section that \emph{every} such class of CQs has an NP-hard fitting problem (cf.~Theorem~\ref{thm:fitting-is-hard}). We observe here that
every such class of CQs lacks the polynomial-size fitting property:

\begin{theorem}
\label{thm:nopolyfit}
Fix a schema $\mathbf{S}$ that contains at least a binary and a unary predicate,
and let $C$ be any class of unary CQs over $\mathbf{S}$ that includes all path-CQs. Then $C$  lacks the polynomial-size fitting property.
\end{theorem}

\begin{proof}
Let $R\in \mathbf{S}$ be binary and $P\in\mathbf{S}$ unary.
For $m\geq 1$, let $L_m$ denote the ``lasso'' instance, with active domain 
$a^m_0, \ldots, a^m_{2m-1}$
consisting of the facts
$R(a^m_i, a^m_{i+1})$ for all $i<2m-1$ and $R(a^m_{2m-1},a^m_m)$ and $P(a^m_m)$.

For $i\geq 1$, let $p_i$ be the $i$-th prime number (where $p_1 = 2$). 
By the prime number theorem, $p_i = O(i\log i)$.

Finally, for $n \geq 1$, let $I_n$ be the disjoint union of $L_{p_i}$ for $i=1, \ldots, n$,
extended with the fact $R(b,b)$ for a fresh value~$b$.
We now construct our set of examples $E_n$ as follows:
\begin{itemize}
    \item Positive example $(I_n,a^{p_i}_0)$ for $i=1\ldots n$.
    \item Negative example $(I_n,b)$.
\end{itemize}

It is easy to see that a fitting path-CQ for $E_n$ exists, namely
the query \[q(x_1)\coloneq \exists x_2\ldots x_k (R(x_1,x_2)\land\cdots\land R(x_{k-1},x_k)\land P(x_k))\] where $k = \Pi_{i=1\ldots n} (p_i)$.

We claim that every CQ that fits the examples must be of size at least
$2^n$.
Let $q(x)$ be any CQ that fits the examples. Since positive and negative examples are based on the same instance, we may assume
that $q$ is connected.
First of all, note that $q$ must contain a conjunct of the form $P(y)$
(otherwise it would fail to fit the negative example). Furthermore, $y$
is not the free variable $x$ and $q$ uses only the relation symbols
$P$ and $R$
(otherwise it would fail to fit any positive example).
Consider the 
directed graph where the vertices are the variables of $q$ and there is an edge from variable $z$ to variable $z'$ iff the atom $R(z, z')$ occurs
in $q$. Since $q$ is connected, there is an undirected
path connecting $x$ to $y$. Take any such path of minimal length. We can represent it as a sequence
\[ x = x_0,\alpha_0, x_1, \alpha_1, \ldots, x_\ell = y\]
where for each $i< \ell$, $\alpha_i$ is an atom that occurs in $q$ that is either $R(x_i, x_{i+1})$ 
(then $\alpha_i$ \revnote{is} a ``forward edge'') or
$R(x_{i+1},x_i)$ (then $\alpha_i$ \revnote{is} a ``backward edge'').
We define the net-length of this path to be the number of forward edges minus the number of backward edges. 

Clearly, in order for the query $q$ to be satisfied
in a lasso instance $L_m$, the net length of the above
path must be divisible by $m$.
Therefore, since $q$ fits all
the examples constructed above, the net-length must be divisible by $p_i$, for all $i=1\ldots n$, and thus at least $\prod_{i=1\ldots n} (p_i)$.
It follows, then, that also the length (in the ordinary sense)
of the path must be at least $\prod_{i=1\ldots n} (p_i)$.
Therefore, every CQ that fits the above examples must have at least
$\prod_{i=1\ldots n} (p_i)$ variables, which exceeds $2^n$. 
\end{proof}

\section{Failure of Strong PAC Learnability}

By Proposition~\ref{prop:strong-learnability}, Theorem~\ref{thm:nopolyfit} implies:

\begin{corollary}\label{cor:non-strong-CQ}
Fix any schema $\mathbf{S}$ that contains at least a
binary relation symbol and a unary relation symbol. Let $C$ be any class of unary CQs over $\mathbf{S}$ that includes all path-CQs.
Then $C$ is not strongly efficiently
PAC learnable. 
\end{corollary}

Alternatively, Corollary~\ref{cor:non-strong-CQ} can be
shown using a VC-dimension argument. In fact, we may then
even drop the `efficiently' from the statement. We define
\emph{strong PAC learnability} in the same way as
\emph{strongly efficient PAC learnability} (cf.~Definition~\ref{def:pac}) except that $A$ is not required to run in polynomial time and $f$ is 
not required to be a polynomial function. 


\begin{theorem}
Fix any schema $\mathbf{S}$ that contains at least a
binary relation symbol and a unary relation symbol. Let $C$ be any class of unary CQs over $\mathbf{S}$ that includes all path-CQs. Then $C$ is 
not strongly PAC learnable.
\end{theorem}

\begin{proof}
Let us recall the definition of VC-dimension. We say that a concept class $C$ \emph{shatters} a set of examples $S$ if for every subset $S'\subseteq S$ there is a $c\in C$ such that $S' = \{e\in S\mid e\models c\}$. The \emph{VC-dimension} of $C$ is the cardinality of the largest set of examples that is shattered by $C$, or infinite if arbitrarily large sets can be shattered. The fundamental theorem of statistical machine learning says that a concept class is strongly PAC learnable iff it has finite VC dimension~\cite{Blumer89:learnability}.

Let $\mathbf{S}$ be a schema that contains a unary relation symbol $P$ and a binary relation symbol $R$, and let $C$ be a class of unary CQs over $\mathbf{S}$ that contains all path-CQs. We show that $C$ has infinite VC-dimension. 

Let $n>0$. We construct a set $S$ that contains $n$ examples $(I_1,a_1),\ldots,(I_{n},a_1)$. Each instance $I_i$ contains an $R$-path of length $n-1$ starting at $a_1$, that is, $\adom(I_i)=\{a_1,\ldots,a_{n}\}$ and $R(a_j,a_{j+1})\in I_i$ for all $j\in\{1,\ldots,n-1\}$. Moreover, we include in $I_i$ all facts $P(a_j)$ for $j\neq i$.

To show that $C$ shatters $S$, let $S'\subseteq S$ be an arbitrary subset of $S$ and let $X\subseteq \{1,\ldots,n\}$ be such that $S' = \{(I_i,a_1)\in S\mid i\in X\}$ and set $\overline X = \{1,\ldots,n\}\setminus X$. Let 
$q(x_1)$ be the path-CQ 
\[
  q(x_1)\coloneq\exists x_2\ldots x_n(\bigwedge_{i=1\ldots n-1}\!\!\!\!\!\! R(x_i,x_{i+1}) ~\wedge \bigwedge_{j\in \overline X} P(x_j)). 
\]
One may verify that $S' = 
 \{(I_i,a_1)\in S\mid q(a_1)\in I_i\}$. 
%
\end{proof}

The concept class of 
path-CQs \emph{is} polynomial-time evaluable, as follows from the fact that it forms a subclass of the class of $\alpha$-acyclic CQs, which is polynomial-time evaluable~\cite{Yannakakis81}. We make use of this in the next section.

\begin{theorem}[\cite{Yannakakis81}]
\label{thm:pathCQpoleval}
Fix any schema $\mathbf{S}$.
The concept class of path-CQs over $\mathbf{S}$ is polynomial-time evaluable.
\end{theorem}

\section{Non-Efficient PAC Learnability}

We now consider PAC learnability in the non-strong version
and show that \emph{no class of unary CQs that includes all path-CQs is
efficiently PAC learnable}, cf.~Theorem~\ref{thm:main} from the introduction. 

Recall that we cannot use Proposition~\ref{prop:pac-vs-fitting} directly to prove non-efficient
PAC learnability, for two reasons. First, the polynomial-size fitting property does not hold for path-CQs. And second, the classes that we consider may contain CQs that are not path-CQs, and thus polynomial-time evaluability also fails, despite Theorem~\ref{thm:pathCQpoleval}. To circumvent the latter issue, we work with a restricted class of instances.

\subsection{Tree-Shaped Instances}

\begin{definition}[Tree-Shaped Instances and CQs]
\label{def:tree-shaped}
Let $\mathbf{S}$ be a schema that consists of a binary relation symbol $R$ and any number of unary relation symbols, and let $I$ be an $\mathbf{S}$-instance. We say that $I$ is 
\emph{tree-shaped} if the following two conditions hold:
\begin{enumerate}
    \item There is a function $\mathit{level} \colon \adom(I)\to \mathbb{N}$ such 
        that, for each fact $R(a,b)$ of $I$,  $\mathit{level}(b)=\mathit{level}(a)+1$.
     \item $I$ does not contain two binary facts $R(a,b), R(a',b)$ that agree on the second value but not on the first.
\end{enumerate}
A CQ over $\mathbf{S}$ is said to be tree-shaped if its canonical 
instance is tree-shaped.%
\footnote{The \emph{canonical instance} of a CQ is
 the instance whose active domain consists 
of the variables of the query and whose facts are
the conjuncts of the query.}
\end{definition}

\begin{lemma}\label{lem:tree-shaped}
Fix a schema $\mathbf{S}$ that consists of one binary relation symbol and any number of unary relation symbols. Given a
CQ $q$ over $\mathbf{S}$, 
\begin{enumerate}
    \item we can test in polynomial time whether there exists a tree-shaped instance $I$ such that $q(I)\neq\emptyset$,
    \item if the answer to the above question is positive, then we can 
construct in polynomial time a
tree-shaped CQ $q'$ such that for all tree-shaped
instances $I$, $q(I)=q'(I)$.
\end{enumerate}
\end{lemma}
    
\begin{proof}
It suffices to prove the claim for connected CQs (the general case
then follows by a component-wise analysis). Therefore, let $q$ 
be a connected CQ.


Let $\sim$ be the 
smallest equivalence relation over the variables of $q$ such that, 
whenever $R(u,v)$ and $ R(u',v')$ are conjuncts of $q$ and $v\sim v'$
then also $u\sim u'$. Let $q'$ be the quotient of $q$ w.r.t.~$\sim$ (that is, $q'$ is obtained from $q$ by choosing a representative
of each $\sim$-equivalence class, and replacing every occurrence of a variable $x$ by the representative of the
$\sim$-equivalence class of $x$).
It is easy to see that, for all tree-shaped instances $I$, $a\in q(I)$ iff $a\in q'(I)$ (here, 
the left-to-right direction uses the tree-shape of $I$, 
while the right-to-left direction holds for every instance $I$).

If $q'$ contains a directed cycle, then clearly, $q'(I)=\emptyset$
for all tree-shaped instances $I$, and we are done.

Assume, therefore, that $q'$ does not contain a directed cycle.
Since
$q'$ is connected, there must then exist a (free or existentially
quantified) variable $y$ for which $q'$ does not 
contain any conjunct of the form $R(\cdot,y)$. 
Furthermore, any simple path from $y$ to any other
variable $z$ must consist entirely of forward edges,
otherwise, the path would be of the form 
\[y\xrightarrow{R} \cdots \xrightarrow{R} u \xrightarrow{R} v \xleftarrow{R} w \xleftarrow{R} \cdots \xleftarrow{R} z\]
and then $u$ and $w$ would have been 
identified when we constructed $q'$.
It follows that $q'$ is tree-shaped.
Furthermore, let $I_{q'}$ be the canonical instance of~$q'$.
Then, clearly, $q'(I_{q'})\neq \emptyset$.
\end{proof}

Since tree-shaped CQs are $\alpha$-acyclic
and hence can be evaluated in polynomial time
(on the class of all instances)~\cite{Yannakakis81}, Lemma~\ref{lem:tree-shaped} immediately implies:

\begin{proposition}\label{prop:ptime-evaluability-trees}
Fix a schema $\mathbf{S}$ that contains one binary relation symbol and any number of unary relation symbols. 
For every class $C$ of CQs over $\mathbf{S}$, the concept class \mbox{$(C,\mathit{Ex}_{\mathit{tree}},\models)$}, where $\mathit{Ex}_{\mathit{tree}}$ is the set of tree-shaped $\mathbf{S}$-instances, is polynomial-time evaluable.
\end{proposition}

In what follows, we will therefore only work with tree-shaped instances.

\subsection{A reduction from 3CNF satisfiability}
\label{sec:reduction1}

Fix a schema $\mathbf{S}$ containing a binary relation symbol $R$
and a unary relation symbol $P$. 

We  use  a reduction from the satisfiability problem for 3CNF formulas, inspired by~\cite{DBLP:conf/ecml/Kietz93,Haussler89}.
Let $\phi=\phi_1\land\cdots\land \phi_k$ be any 3CNF formula over a propositional
signature $\mathit{PROP}=\{X_1, \ldots, X_m\}$.  We denote by $\mathit{LIT}=\{X_i, \overline{X_i}\mid i\leq m\}$ the set of all literals over $\mathit{PROP}$. For every  $l \in \mathit{LIT}$, set 
    $j_l = 2i$ if $l$ is of the form $X_i$ and $j_l=2i-1$ if $l$ is of the form $\overline{X_i}$.
Define an $\mathbf{S}$-instance $I_\phi$ as follows:
\begin{itemize}
    \item $R(a_i,p_{i,1})$ and $R(a_i,n_{i,1})$ for $i\leq m$
    \item $R(p_{i,j}, p_{i,j+1})$ and $R(n_{i,j}, n_{i,j+1})$ for $i\leq m$, $j< 2m$
    \item $P(p_{i,j_l})$ for every literal $l\in \mathit{LIT}\setminus\{\overline{X_i}\}$
\item  $P(n_{i,j_l})$  for every literal $l\in \mathit{LIT}\setminus\{X_i\}$
    \item $R(b,b_{i,1})$ for $i\leq k$
    \item $R(b_{i,j},b_{i,j+1})$ for $i\leq k$ and $b\leq 2m$
    \item $P(b_{i,j_l})$ for every $l\in \mathit{LIT}$ and $i\leq k$ with $l$ not occurring in the clause $\phi_i$.
\end{itemize}

    Let $E_\phi=\{((I_\phi,a_i),+)\mid i\leq m\}\cup \{((I_\phi,b),-)\}$.

\begin{example}
    Let $\mathit{PROP}=\{X_1, X_2\}$ and consider the formula $\phi = X_1 \land X_2 \land (\overline{X_1}\lor X_2)$. Then, the corresponding $\mathbf{S}$-instance $I_\phi$ can be depicted as follows (where each edge represents an $R$-edge directed downwards):
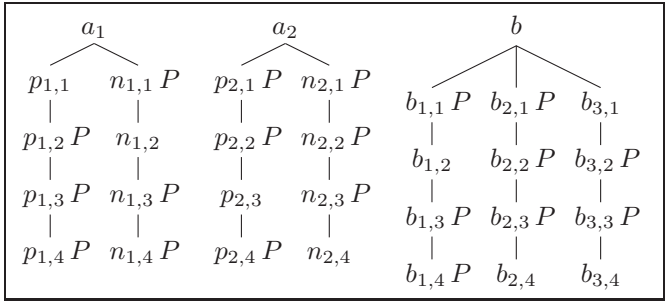
\begin{center}
\fbox{
\Tree [.{$a_1$} 
  [.{$p_{1,1}$} [.{$p_{1,2} \, P\!\!\!$}  [.{$p_{1,3} \, P\!\!\!$} [.{$p_{1,4} \, P\!\!\!$} ] ] ] ]  
  [.{$n_{1,1} \, P\!\!\!$} [.{$n_{1,2}$}  [.{$n_{1,3} \, P\!\!\!$} [.{$n_{1,4} \, P\!\!\!$} ] ] ] ]  
]
~~~~
\Tree [.{$a_2$} 
  [.{$p_{2,1} \, P\!\!\!$} [.{$p_{2,2} \, P\!\!\!$}  [.{$p_{2,3}$} [.{$p_{2,4} \, P\!\!\!$} ] ] ] ]  
  [.{$n_{2,1} \, P\!\!\!$} [.{$n_{2,2} \, P\!\!\!$}  [.{$n_{2,3} \, P\!\!\!$} [.{$n_{2,4}$} ] ] ] ]  
]
~~~~
\Tree [.{$b$} 
  [.{$b_{1,1} \, P\!\!\!$} [.{$b_{1,2}$}  [.{$b_{1,3} \, P\!\!\!$} [.{$b_{1,4} \, P\!\!\!$} ] ] ] ]  
  [.{$b_{2,1} \, P\!\!\!$} [.{$b_{2,2} \, P\!\!\!$}  [.{$b_{2,3} \, P\!\!\!$} [.{$b_{2,4}$} ] ] ] ]  
  [.{$b_{3,1}$} [.{$b_{3,2} \, P\!\!\!$}  [.{$b_{3,3} \, P\!\!\!$} [.{$b_{3,4}$} ] ] ] ]  
]
~~
}
\end{center}
\end{example}

\begin{lemma}\label{lem:reduction}\ 
    For all 3CNF formulas $\phi$:
\begin{enumerate}
\item 
From a satisfying assignment for $\phi$, one can 
construct in polynomial time a path-CQ that fits $E_\phi$. 
\item 
Conversely, if there is a CQ that fits $E_\phi$, then $\phi$ has
a satisfying assignment.
\end{enumerate}
In particular, whenever there is a CQ that fits $E_\phi$, then there is a fitting path-CQ
of size polynomial in $|\mathit{PROP}|$.
\end{lemma}

\begin{proof}
1. Let $v$ be a satisfying assignment for $\phi$. Let
\[\begin{array}{l}
    q(x_0) \coloneq \exists x_1, \ldots x_{2m} (R(x_0, x_1) \land\cdots\land R(x_{2m-1},x_{2m}) \land {} \\ \hspace{40mm} \bigwedge_{l\in \mathit{LIT} \text{ such that } v\models l} P(x_{j_l})).
\end{array}\]
Clearly, each $a_i\in q(I_\phi)$ and $b\not\in q(I_\phi)$.

2. Let $q(x)$ be a unary CQ that fits $E_\phi$. 
By Lemma~\ref{lem:tree-shaped}, we may assume
that $q$ is a tree-shaped CQ. Furthermore,
we may assume without loss of generality that
$q$ is connected.
Let $\mathit{level}_q \colon \mathit{Vars}(q)\to\mathbb{N}$ be as given by 
Definition~\ref{def:tree-shaped}. We may 
assume $\mathit{level}_q(x)=0$ (if there was any $y\in \mathit{Vars}(q)$
with $\mathit{level}_q(y)<\mathit{level}_q(x)$, then $q$ would not fit
the positive examples of $E_\phi$).

Thus, $q(x)$ is a connected tree-shaped CQ, where $x$ is 
the root of the tree. Since $q(x)$ fits the negative
example $(I_\phi,b)$,
we have that $b\not\in q(I_\phi)$. This means
that either (i) $q$ contains a conjunct of the 
form $P(x)$, or (ii) for some $y\in \mathit{Vars}(q)$ with 
$\mathit{level}_q(y)=1$, the subtree of $q$ rooted at
$y$, does not admit a homomorphism to $(I_\phi,b_{i,1})$ for any $i\leq n$. It is easy
to see that (i) cannot happen, because it would
imply that $q$ does not fit the positive
examples in $E_\phi$. Therefore, case (ii) must
apply. Let $y$ be the variable in question, and 
let us denote by $q'(y)$ the subtree of $q$ rooted
at $y$ (with $y$ as its free variable). 

We know that $q'(I_\phi)$ does not contain 
$b_{i,1}$ for any $i\leq n$. Furthermore,
it is easy to see (from the fact that $q$ fits the positive examples in $E_\phi$), that
for each $i\leq m$, either $p_{i,1}$ or $n_{i,1}$ belongs
to $q'(I_\phi)$.

Now, 
let $L_y$ be the set $$\{l\in \mathit{LIT} \mid \text{$q'$ has a conjunct $P(z)$ with $\mathit{level}_q(z)=j_l+1$}\}$$ 

\par\noindent\textbf{Claim 1:} $L_y$ does not contain both
$X_i, \overline{X_i}$ for any $i\leq m$.

\medskip

Claim 1 follows immediately from the fact that $q(x)$ fits the 
positive examples.

\medskip

\par\noindent\textbf{Claim 2:} $L_y$ contains a
literal from each clause of $\phi$.

\medskip

Suppose, for the sake of a contradiction, that $\phi$ has a
clause $\phi_i$, such that no literal occurring in $\phi_i$
belongs to $L_y$. Then, $b_{i,1}$ belongs to $q'(I_{\phi})$,
as witnessed by the variable assignment that maps 
each variable $z$ to $b_{i,\mathit{level}_{q}(z)-1}$. However, 
we know that $b_{i,1}\not\in q'(I_{\phi})$, a contradiction.

Claim 1 and 2 together imply that $\phi$ is 
satisfiable. Indeed, it suffices to take any truth assignment consistent with the literals in $L_y$.
\end{proof}

From Lemma~\ref{lem:reduction}, together with the NP-hardness of 3CNF satisfiability, we immediately get:

\begin{theorem}\label{thm:fitting-is-hard}
Fix any schema $\mathbf{S}$ that contains at least a binary relation symbol and a unary relation symbol,
and let $C$ be any class of unary CQs over $\mathbf{S}$ that includes all path-CQs. Then the fitting problem for 
$C$ is NP-hard.
\end{theorem}

Now, putting everything together, we can prove
Theorem~\ref{thm:main}, restated here:

\begin{theorem}\label{thm:main1} (assuming $\text{RP}\neq \text{NP}$)
Fix a schema $\mathbf{S}$ containing at least one binary 
relation symbol $R$ and one unary relation symbol $P$. Let
$C$ be any class of unary CQs over $\mathbf{S}$ that includes
all path-CQs. 
Then $C$ is not 
efficiently PAC learnable, even w.r.t. single-instance distributions. 
\end{theorem}

\begin{proof}
    Assume that the concept class $C=(C,\mathit{Ex},\models)$ is 
efficiently PAC learnable. Then, by Proposition~\ref{prop:subspace}, so is
$C'=(C,\mathit{Ex}',\models)$ where $\mathit{Ex}'=\{(I,a)\mid I=I_\phi 
    \text{ for some 3CNF formula } \phi \text{ and }
    a \in \{a_1,a_2,b\}\}$. 
It follows from Lemma~\ref{lem:reduction} that $C'$ has
the polynomial-size fitting property. 
Furthermore,  $C'$
is polynomial-time evaluable  since all examples in $\mathit{Ex}'$ are tree-shaped and by Proposition~\ref{prop:ptime-evaluability-trees}. By Proposition~\ref{prop:pac-vs-fitting}, the fitting problem for $C'$ is thus solvable in RP. 
By Theorem~\ref{thm:fitting-is-hard}, this implies that
$\text{RP} = \text{NP}$.

A careful inspection of the proof of Proposition~\ref{prop:pac-vs-fitting}
and the construction of our examples reveals that even efficient PAC learnability w.r.t.~single-instance distributions already
gives us, in the same way as above, an RP-algorithm for the fitting problem for $C'$.
\end{proof}

\begin{remark}\rm
The above proof involves path-CQs
of unbounded depth, over a fixed schema. It is easy to see that if we were to bound both the 
depth of the path-CQs and keep the schema fixed,  we would end up with a finite concept class, trivializing the PAC learning problem.
\end{remark}

\begin{remark}\rm
The above non-learnability proof cannot be adapted to
 UCQs in an obvious way. In fact, we crucially use the fact that
the fitting problem for path-CQs is NP-hard whereas the 
fitting problem for UCQs that are unions of path-CQs can be solved in
polynomial time. On the other hand, as mentioned earlier, 
it follows from results in \cite{CateDK13:learning} that UCQs are not efficiently PAC learnable, assuming $\text{RP} \neq \text{NP}$.
\end{remark}

\begin{remark}\rm
The fact that the above proof involves a reduction from the satisfiability problem for 3CNF formulas is remarkable, 
given that 3CNF formulas themselves \emph{are}  efficiently PAC learnable~\cite{KearnsVazirani}.
\end{remark}

\begin{remark}\rm
Efficient PAC learnability as defined in
Definition~\ref{def:pac} (in the non-strong version) is sometimes also known as 
\emph{strong PAC learnability}. In contrast, 
\emph{weak PAC learnability} then merely requires the
existence of a learner that works for 
\emph{some} non-trivial choice of $\delta$ and $\epsilon$.
A well-known result in computational learning theory
states that, for polynomial-time evaluable concept
classes, weak learnability implies strong learnability (cf.~\cite{KearnsVazirani}).
Since the concept class of CQs is not polynomial-time
evaluable, Theorem~\ref{thm:main}, taken at face value,
does not imply that the same result holds in the weak PAC model.
Nevertheless, inspection of our proof immediately shows that it yields the same result also for the weak
PAC model. 
\end{remark}

\section{PAC Learnability with Membership Queries}

We prove Theorem~\ref{thm:main2} from the introduction.
Formally, a \emph{membership oracle} $\mathit{MEMB}_{\phi}$, for a 
concept~$\phi$, is an oracle that, given any unlabeled
example $e$, returns (in unit time) its label according to $\phi$.
PAC learning with access to a membership oracle for the target
concept can be
viewed as a formal model of active learning.

\begin{theorem}\label{thm:occam-membership}
Fix any schema $\mathbf{S}$ and $k\geq 0$. There is an algorithm that takes
as input a set $E$ of examples labeled according to a $k$-ary CQ $q^*$ over $\mathbf{S}$, has access to a membership oracle for $q^*$,
and outputs a $k$-ary CQ $q$ over $\mathbf{S}$ with 
$|q| \leq |q^*|$ that fits $E$. Moreover, the running time of the algorithm 
is polynomial in $||E||$ and~$|q^*|$.
\end{theorem}

\begin{proof}
We use ideas similar to the ones used in the proof that CQs are efficiently exactly learnable with membership and equivalence queries~\cite{CateDK13:learning,tCD2022:conjunctive}.
Before we describe the algorithm, we introduce a number of basic concepts. 

Let $I,J$ be instances over the same schema. A mapping $h\colon \adom(I)\to \adom(J)$ is called \emph{homomorphism from $I$ to $J$} if $R(h(\textbf{c}))\in J$ for every $R(\textbf{c})\in I$. Given tuples $\textbf{a}$ and $\textbf{b}$ of values from $I$ and $J$, respectively, we write
$(I,\textbf{a})\to (J,\textbf{b})$ to denote the existence of a 
homomorphism $h$ from $I$ to $J$ with $h(\textbf{a})=\textbf{b}$. Homomorphisms compose in the sense that $(I,\textbf{a})\to (J,\textbf{b})$ and $(J,\textbf{b})\to (K,\textbf{c})$ implies $(I,\textbf{a})\to (K,\textbf{c})$.

The \emph{direct product} $I\times J$ of two instances (over the same
schema $\textbf{S}$), is the $\textbf{S}$-instance that consists of
all facts of the form $R(\langle a_1, b_1\rangle, \ldots, \langle a_n,b_n\rangle)$,
where $R(a_1, \ldots, a_n)$ is a fact of $I$ and 
$R(b_1, \ldots, b_n)$ is a fact of $J$. Note that the active 
domain of $I\times J$ consists of pairs from $\adom(I)\times \adom(J)$.
The direct product $(I,\textbf{a})\times (J,\textbf{b})$ of two examples, where $\textbf{a}=a_1, \ldots, a_k$ and $\textbf{b}=b_1, \ldots b_k$ are of the same length, is given by $(I\times J, (\langle a_1, b_1\rangle, \ldots, \langle a_k,b_k\rangle)$. Note that, in general, this may not yield a well-defined example, because there is no guarantee that the distinguished elements $\langle a_1, b_1\rangle, \ldots, \langle a_k,b_k\rangle$
belong to $\adom(I\times J)$. When it is well-defined, then the projections to the respective components witness that both $(I,\textbf{a})\times (J,\textbf{b}) \to (I,\textbf{a})$ and $(I,\textbf{a})\times (J,\textbf{b}) \to (J,\textbf{b})$.
 
A \emph{critical positive example} for a CQ $q^*$ is a 
positive example $(I,\textbf{a})$ for $q^*$, such that,
for every proper subinstance $I'\subsetneq I$, 
$(I',\textbf{a})$ is a negative example for $q^*$.

The following claim is easy to prove
(\cite[Lemma 5.4]{tCD2022:conjunctive}):

\medskip\par\noindent\textbf{Claim 1:}
Given a positive example $(I,\textbf{a})$ for an unknown CQ $q^*$,
we can construct from it in linear time a critical positive example
$(I',\textbf{a})$ for $q^*$, with $I'\subseteq I$, given access to a membership oracle for $q^*$.

\medskip\par\noindent\textbf{Claim 2:}
If $(I,\textbf{a})$ and $(J,\textbf{b})$ are positive examples for
a CQ $q^*$, then $(I,\textbf{a})\times (J,\textbf{b})$ is a well-defined example, and it is a positive example for $q^*$.

\emph{Proof of Claim 2:}
Let $(I,\textbf{a})$ and $(J,\textbf{b})$ be positive examples for a CQ $q^*$. Let
$h_1$ and $h_2$ be the respective witnessing variable assignments. Then the map $h$ given
by $h(x)=(h_1(x),h_2(x))$ is a satisfying variable
assignment for $q^*$ in $(I,\textbf{a})\times (J,\textbf{b})$, showing that the latter is a positive
example for $q^*$. It remains to show that it is a 
well-defined example, i.e., that each distinguished element occurs in a fact. This follows 
from the fact that each free variable of $q^*$ occurs
in a conjunct of $q^*$ (by the definition of CQs),
and that each distinguished element of $(I,\textbf{a})\times (J,\textbf{b})$ is the $h$-image of a free variable of $q^*$ (cf.~\cite[Lemma 5.5]{tCD2022:conjunctive}).
\medskip

Given a set $E$ of examples labeled according to~$q^*$, the algorithm proceeds as follows. Let $(I_1, \textbf{a}_1), \ldots, (I_n,\textbf{a}_n)$ be an enumeration of the positive examples in $E$.
We construct, by induction on $n$, a critical positive example
$(J,\textbf{b})$ for $q^*$ such that there is a homomorphism from
$(J,\textbf{b})$ to each $(I_i, \textbf{a}_i)$.
This is done by applying Claim~1 and Claim~2 in an interleaved 
fashion. More precisely: 
\begin{itemize}

    \item Start by setting $(J_1,\textbf{b}_1)$ to be the critical positive example obtained from $(I_1,\textbf{a}_1)$ via Claim~1.
    
    \item For $i=2,\ldots,n$, let $(J_i',\textbf{b}_i')$ be  $(J_{i-1},\textbf{b}_{i-1})\times (I_{i},\textbf{a}_i)$ and obtain $(J_i,\textbf{b}_i)$ as critical positive example from $(J'_i,\textbf{b}'_i)$ via Claim~1.
    \item Set $(J,\textbf{b})=(J_n,\textbf{b}_n)$.
\end{itemize}
Note that, by Claim~2 and the fact that homomorphisms compose, each $(J_i',\textbf{b}_i')$ is a well-defined example that has a homomorphism to all examples $(I_1,\textbf{a}_1),\ldots,(I_i,\textbf{a}_i)$. Thus, $(J,\textbf{b})$ has a homomorphism to all positive examples. 
Let $\textbf{b} = b_1,\ldots,b_k$ and let $q$ be the canonical CQ of $(J,\textbf{b})$, that is,
the CQ $q(x_{b_1}, \ldots, x_{b_k})$ that has a conjunct
for every fact of $J$, where each element $b\in \adom(J)$ 
is replaced by a corresponding variable $x_{b}$. 
Then $q$ fits the positive examples in $E$
since $(J,\textbf{b})$ has a homomorphism to each positive
example. It also fits the negative examples in~$E$:
$(J,\textbf{b})$ is a positive example for $q^*$ by construction  and if $q$ fails to fit a negative example $(I,\textbf{a})$ in $E$, then
 $(J,\textbf{b})$ has a homomorphism to $(I,\textbf{a})$, which, by composition of homomorphisms,  leads to a contradiction with 
 $q^*$ fitting~$(I,\textbf{a})$.
 
Furthermore, one can easily see that any 
\emph{critical} positive example  $(I,\textbf{a})$ for $q^*$ satisfies $|I|\leq |q^*|$. Hence,
each $J_i$ satisfies $|J_i|\leq |q^*|$. This implies, in particular, that $|q|\leq |q^*|$ as required. Moreover, it implies that $|J_i'|\in O(||E||\cdot |q^*|)$, for all $i$. Since $J_i$ is obtained from $J_i'$ in linear time by Claim~1, the running time of this algorithm is $O(||E||^2\cdot |q^*|)$.
\end{proof}

The algorithm given in Theorem~\ref{thm:occam-membership} is 
an Occam algorithm (with $\alpha=0$ and $k=1$) in the sense of Definition~\ref{def:occam}, except for the 
fact that it uses a membership oracle. While Theorem~\ref{thm:occam} is stated for the case without
membership queries, its proof applies also to Occam algorithms with membership queries, yielding efficient PAC learnability with membership queries (stated as Theorem~\ref{thm:main2} in the introduction):

\begin{corollary}
Fix any schema $\mathbf{S}$ and $k\geq 0$. The class of all $k$-ary CQs over $\mathbf{S}$ is efficiently PAC learnable with membership queries.
\end{corollary}

\begin{remark}\rm\label{rem:prediction-vs-learning}
The proof of Theorem~\ref{thm:occam-membership} establishes something stronger, namely that CQs are
efficiently PAC learnable with membership queries even when the schema 
$\mathbf{S}$ and the arity $k$ are not fixed but treated as part of the input of the learning task. This is remarkable, because it follows
from results in~\cite{Dalmau1999} that CQs are not PAC predictable with membership queries when the arity is treated
as part of the input (under suitable cryptographic assumptions). However, note that efficient PAC learnability (with membership queries) implies PAC predictability (with membership queries) only for concept classes that are polynomial-time evaluable, which the class of CQs is not.
\end{remark}

\begin{remark}\rm
We expect that, with respect to each of the various notions of ``acyclicity'' mentioned in the introduction, acyclic CQs are efficiently PAC learnable with membership queries. 
However, since efficient PAC learnability (with or without membership queries) is not a monotone property of concept classes, this requires a case-by-case analysis. A challenge
is posed by the fact that the positive examples $(I_i,\mathbf{a}_i)$ are not guaranteed to correspond to queries from the considered class, and thus neither are the hypotheses that our algorithm generates.
\end{remark}

The above proof can also be modified to apply to the concept class of \emph{unions of conjunctive queries (UCQs)}. By a $k$-ary UCQ over a schema $\textbf{S}$ we mean a non-empty finite disjunction of $k$-ary CQs over $\textbf{S}$.

\begin{theorem}\label{thm:ucq}
Fix any schema $\mathbf{S}$ and $k\geq 0$. The class of $k$-ary UCQs over $\mathbf{S}$ is efficiently PAC learnable with membership queries.
\end{theorem}

\begin{proof} 
We sketch the modified algorithm.
Given a set $E$ of labeled examples, it  proceeds as follows. Let $(I_1, \textbf{a}_1), \ldots, (I_n,\textbf{a}_n)$ be an enumeration of the positive examples in $E$.
We construct 
sets of critical positive examples $X_0,\dots,X_n$ such that for all $i$ and all $(I_j, \textbf{a}_j)$ with $j < i\leq n$, there exists a $(J,\textbf{b})\in X_i$ that admits a homomorphism  to  $(I_j, \textbf{a}_j)$.
As before, this is done by applying Claim~1 and Claim~2 in an interleaved 
fashion. 

More precisely, set $X_0 = \emptyset$;
for $i=1, \ldots, n$,
        we first test whether there is a $(J,\textbf{b})\in X_{i-1}$ such that $(J,\textbf{b})\times (I_i,\textbf{a}_i)$ is a positive example
        for the target query $q^*$. We use a membership 
        query for this. If such $(J,\textbf{b})\in X_{i-1}$ exists, then we choose an arbitrary one
        and set $X_i = (X_{i-1}\setminus \{(J,\textbf{b}\})\cup \{(J',\textbf{b}')\}$,
        where $(J',\textbf{b}')$ is a subinstance of
        $(J,\textbf{b})\times (I_i,\textbf{a}_i)$ that is a 
        critical positive example for~$q^*$. 
        Otherwise (if no such $(J,\textbf{b})\in X_{i-1}$ exists), we set $X_i = X_{i-1}\cup \{ (J',\textbf{b}')\}$
        where $(J',\textbf{b}')$ is a subinstance of
        $(I_i,\textbf{a}_i)$ that is a 
        critical positive example for $q^*$.

Let $q$ be the UCQ that is the disjunction of the canonical CQs of the examples in $X_n$.
By similar arguments as before, we can show that $q$ fits $E$ and  $|q|\leq |q^*|$.
In particular,  for each $i\leq n$ the sum of the sizes of the structures in
$X_i$ is at most the size of $q^*$. 
\end{proof}
\begin{remark}\rm
The problem of learning \emph{GAV schema mappings}
closely corresponds to the problem of learning UCQs
(cf.~\cite{tCD2022:conjunctive}). In particular,
Theorem~\ref{thm:ucq} 
implies that GAV schema mappings are efficiently PAC learnable with membership queries.
This
resolves an open question in~\cite{CateDK13:learning}.
\end{remark}

\section{Conclusion}

We established a strong negative result on the efficient PAC learnability of classes of CQs that include all path-CQs. Although our result indicates that interesting
classes of CQs tend to not be efficiently PAC learnable,
from a theoretical perspective it would be interesting to work  towards a complete classification of classes of CQs that are (or are not) efficiently PAC learnable. On the positive side,
we showed that CQs and UCQs are efficiently PAC learnable with membership queries.



In the following, we discuss how one could try to overcome the  negative result by loosening the running time requirements.
%
A first observation is that while PAC learnability of (the class of all) CQs cannot be attained by a polynomial-time algorithm, PAC learning with \emph{only polynomial sample size} is always possible when more running time is granted. 
\revnote{Indeed, this approach has been successfully exploited in 
\cite{ourNewIJCAI} for PAC learning
unary tree-shaped CQs 
(over a schema that contains only unary and binary relations) 
with the help of a SAT solver.}

The fact that a PAC learning algorithm for CQs exists with polynomial sample size but super-polynomial running time, is not difficult 
to establish.
One can simply use an Occam algorithm that enumerates 
candidate CQs $q$ in the order of increasing size,
checks for each $q$ whether it fits the input examples $E$, and returns the first fitting CQ found. 
If a fitting CQ
exists, then there is one of 
size single exponential in $||E||$ \cite{CateD15}. We may thus
terminate (and return an arbitrary CQ) when that bound
is reached. The algorithm runs in double exponential time
even if we check in a brute-force way whether candidate CQs fit the input examples. 
The bound on the sample size stated after Theorem~\ref{def:occam}  applies despite the non-polynomial running time. We thus obtain
 a PAC algorithm with polynomial sample size and double exponential running time.
 
It is an interesting question whether and when a more modest superpolynomial running time suffices. In particular, one may 
consider running times that also depend on the 
target query $q$ rather than only on the input set of examples $E$.
From this perspective, the above algorithm attains running time $||E||^{O(|q|)}$ while time $f(q) \cdot \text{poly}(||E||)$ with $f$ a computable function would clearly be preferable.
This resembles fixed-parameter tractability (FPT) in the study of the parameterized complexity 
of query evaluation (with the size of the query being the parameter), so let us refer to it as \emph{FPT PAC learning}. 
To make this well-defined, it is convenient to view FPT PAC learning as a promise problem, meaning that the input examples are promised to have a fitting query from the considered class.\footnote{The lower bounds proved in this paper then no longer apply.} Alternatives are to treat the non-existing target query as being of size~1 (which is a strong requirement) and to grant unlimited running time in the case that there is no fitting query (declaring that case a corner case). 

In the setting of FPT PAC learning, classes of CQs of bounded
treewidth and (more generally) bounded submodular width should be expected to play a prominent role because these notions are tightly linked to CQ evaluation in FPT \cite{DBLP:journals/jacm/Grohe07,DBLP:journals/jacm/Marx13}. They  generalize all notions of acyclicity mentioned in this paper, such as $\alpha$-acyclicity. The exact same Occam algorithm 
described above yields that
for every \mbox{$k \geq 1$}, the class $C_k$ of CQs of submodular width at most $k$ is FPT PAC learnable  with polynomial sample size. 
%
This raises a number of questions: Is the class of all CQs FPT PAC learnable? If not, can we characterize the classes of  CQs that are? And how exactly does the running time of the
algorithms depend on the parameter?

\bibliographystyle{plain}
\bibliography{bib}

\end{document}